\documentclass[onecolumn, a4size, 11pt]{IEEEtran}
\usepackage{amsmath}
\usepackage{amssymb}
\usepackage{amsfonts}
\usepackage{graphicx}
\usepackage{epsfig}
\usepackage{subfigure}
\usepackage{psfrag}
\usepackage{xcolor}
\usepackage{epstopdf}

\title{Downlink SINR Balancing in C-RAN under Limited Fronthaul Capacity\footnote{The authors are with the
Department of Electrical and Computer Engineering, National
University of Singapore (e-mail:liu\_liang@u.nus.edu; elezhang@nus.edu.sg). R. Zhang is also
with the Institute for Infocomm Research, A*STAR, Singapore.}}

\author{Liang Liu ~\IEEEmembership{Member,~IEEE} and Rui Zhang ~\IEEEmembership{Senior Member,~IEEE}}

\begin{document}
\maketitle \thispagestyle{empty} \vspace{-0.3in}

\begin{abstract}
Cloud radio access network (C-RAN) with centralized baseband processing is envisioned as a promising candidate for the next-generation wireless communication network. However, the joint processing gain of C-RAN is fundamentally constrained by the finite-capacity fronthaul links between the central unit (CU) where joint processing is implemented and distributed access points known as remote radio heads (RRHs). In this paper, we consider the downlink communication in a C-RAN with multi-antenna RRHs and single-antenna users, and investigate the joint RRH beamforming and user-RRH association problem to maximize the minimum signal-to-interference-plus-noise ratio (SINR) of all users subject to each RRH's individual fronthaul capacity constraint. The formulated problem is in general NP-hard due to the fronthaul capacity constraints and thus is difficult to be solved optimally. In this paper, we propose a new iterative method for this problem which decouples the design of beamforming and user association, where the number of users served by each RRH is iteratively reduced until the obtained beamforming and user association solution satisfies the fronthaul capacity constraints of all RRHs. A monotonic convergence is proved for the proposed algorithm, and it is shown by simulation that the algorithm achieves significant performance improvement over other heuristic solutions.
\end{abstract}


\setlength{\baselineskip}{1.3\baselineskip}
\newtheorem{definition}{\underline{Definition}}[section]
\newtheorem{fact}{Fact}
\newtheorem{assumption}{Assumption}
\newtheorem{theorem}{\underline{Theorem}}[section]
\newtheorem{lemma}{\underline{Lemma}}[section]
\newtheorem{corollary}{\underline{Corollary}}[section]
\newtheorem{proposition}{\underline{Proposition}}[section]
\newtheorem{example}{\underline{Example}}[section]
\newtheorem{remark}{\underline{Remark}}[section]
\newtheorem{algorithm}{\underline{Algorithm}}[section]
\newcommand{\mv}[1]{\mbox{\boldmath{$ #1 $}}}

\section{Introduction}\label{sec:Introduction}
With dense deployment of distributed access points known as remote radio heads (RRHs) under the coordination of a central unit (CU), cloud radio access network (C-RAN) has been envisioned as a promising candidate for the fifth-generation (5G) wireless networks in future \cite{ChinaMobile}. Unlike the base station (BS) in the traditional cellular networks which encodes or decodes the user messages locally, in C-RAN each RRH merely forwards the signals of wireless users from/to the CU via a high-speed fronthaul link (fiber or wireless) in the downlink and uplink communications, respectively, while leaving the joint encoding/decoding complexity to a baseband unit (BBU) in the CU. The centralized baseband processing at the CU enables enormous spectrum efficiency and energy efficiency gains for C-RAN over conventional cellular networks.

Despite the theoretical performance gains, the practically achievable throughput of C-RAN is largely constrained by the finite-capacity fronthaul links between the RRHs and the CU. In the literature, a considerable amount of effort has been dedicated to study effective techniques to reduce the fronthaul capacity required in both the uplink and downlink communications in C-RAN. In the uplink communication, the so-called ``quantize-and-forward (QF)'' scheme is proposed to reduce the communication rates between the CU and RRHs, where each RRH samples, quantizes and forwards its received wireless signals to the CU over its fronthaul link with a given capacity \cite{Yu13} --\cite{Liu15}. In the downlink communication, besides the QF scheme \cite{Simeone13}, the CU can more efficiently send the user messages to each RRH directly over its fronthaul link, which then encodes the user messages into wireless signals and transmits them to users \cite{Gesbert11}, \cite{Yu14}. In this scheme, user-RRH association is crucial to the performance of C-RAN since in general the CU can only send the messages for a subset of users to each RRH due to the limited capacity of each fronthaul link.

In this paper, we consider the downlink communication in a C-RAN consisting of multi-antenna RRHs and single-antenna users, where user messages are sent from CU to distributed RRHs via individual fronthaul links for coordinated transmission, as shown in Fig. \ref{fig1}. By jointly designing the beamforming at all RRHs and user-RRH associations, we aim to maximize the minimum signal-to-interference-plus-noise ratio (SINR) of all users subject to each RRH's individual transmit power constraint as well as fronthaul capacity constraint. It is worth noting that without the fronthaul capacity constraints, each user can be served by all the RRHs and the resulted beamforming problem for SINR balancing has been solved in \cite{Eldar06}, \cite{Rui10} by utilizing bisection method jointly with conic optimization techniques \cite{Boyd04}. However, with the newly introduced fronthaul constraints, the joint optimization of beamforming and user association results in a combinatorial problem, which is NP-hard and thus difficult to be optimally solved in a network with large number of users and RRHs. In this paper, we propose a new method for practically solving this problem, which effectively decouples the design of RRH beamforming and user-RRH association, thus achieving significant complexity reduction. Specifically, we first associate each user to all RRHs, and then iteratively reduce the number of users served by each RRH until the corresponding optimal beamforming solution given this user association solution satisfies all the RRHs' fronthaul capacity constraints. A monotonic convergence is proved for the proposed iterative algorithm, and numerical results show that its performance is significantly better as compared to other heuristic solutions, especially when the fronthaul capacity is more stringent.

\begin{figure}
\begin{center}
\scalebox{0.5}{\includegraphics*{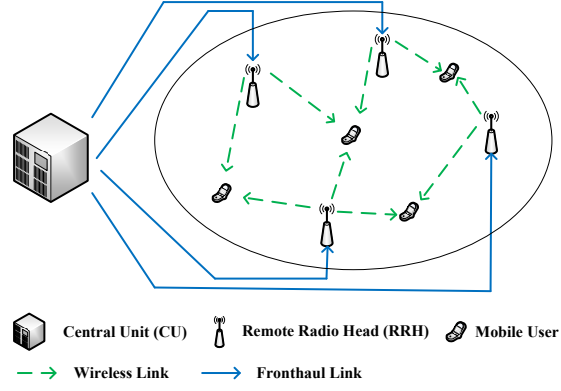}}
\end{center}
\caption{System model of the downlink communication in C-RAN.}\label{fig1}
\end{figure}

\section{System Model}\label{sec:System Model}
This paper studies the downlink communication in C-RAN, as shown in Fig. \ref{fig1}. The studied system consists of one CU, $N$ RRHs, denoted by the set $\mathcal{N}=\{1,\cdots,N\}$, and $K$ users, denoted by $\mathcal{K}=\{1,\cdots,K\}$. It is assumed that each RRH is equipped with $M\geq 1$ antennas, while each user is equipped with one single antenna. It is further assumed that each RRH $n$ is connected to the CU via a fronthaul link with a capacity of $\bar{T}_n$ bits per second (bps). In the downlink, the CU sends the user messages and corresponding quantized beamforming vectors to each RRH via its fronthaul link. Then, each RRH upconverts the digital messages into wireless signals and sends them to the users. The details are given as follows.

It is assumed that the $N$ RRHs communicate with the $K$ users over quasi-static flat-fading channels over a given bandwidth of $B$ Hz. The equivalent baseband transmit signal of RRH $n$ is
\begin{align}\label{eqn:transmit signal scheme 1}
\mv{x}_n=\sum\limits_{k=1}^K\mv{w}_{k,n}s_k, ~~~ n=1,\cdots,N,
\end{align}where $s_k\sim \mathcal{CN}(0,1)$ denotes the message intended for user $k$, which is modeled as a circularly symmetric complex Gaussian (CSCG) random variable with zero-mean and unit-variance, and $\mv{w}_{k,n}\in \mathbb{C}^{M\times 1}$ denotes RRH $n$'s beamforming vector for user $k$. Suppose that RRH $n$ has a transmit sum-power constraint $\bar{P}_n$; from (\ref{eqn:transmit signal scheme 1}), we thus have $E[\mv{x}_n\mv{x}_n^H]=\sum_{k=1}^K\|\mv{w}_{k,n}\|^2\leq \bar{P}_n$, $\forall n$.

Then, the received signal of user $k$ can be expressed as
\begin{align}\label{eqn:received signal scheme 1}
y_k=\sum\limits_{n=1}^N\mv{h}_{k,n}^H\mv{x}_n+z_k=\sum\limits_{n=1}^N\mv{h}_{k,n}^H\mv{w}_{k,n}s_k+\sum\limits_{n=1}^N\mv{h}_{k,n}^H\sum\limits_{j\neq k}\mv{w}_{j,n}s_j+z_k, ~~~ k=1,\cdots,K,
\end{align}where $\mv{h}_{k,n}\in \mathbb{C}^{M\times 1}$ denotes the channel from RRH $n$ to user $k$, and $z_k\sim \mathcal{CN}(0,\sigma^2)$ denotes the additive white Gaussian noise (AWGN) at user $k$. In this paper, it is assumed that the channels to all the $K$ users are perfectly known at the CU.


The decoding SINR for user $k$ is thus expressed as
\begin{align}\label{eqn:SINR scheme 1}
\gamma_k=\frac{\left|\sum\limits_{n=1}^N\mv{h}_{k,n}^H\mv{w}_{k,n}\right|^2}{\sum\limits_{j\neq k}\left|\sum\limits_{n=1}^N\mv{h}_{k,n}^H\mv{w}_{j,n}\right|^2+\sigma^2}, ~~~ k=1,\cdots,K.
\end{align}Then, the achievable rate of user $k$ in bps is given by
\begin{align}
R_k=B\log_2(1+\gamma_k), ~~~ k=1,\cdots,K.
\end{align}

Next, consider the data transmission from the CU to RRHs over their fronthaul links. It is worth noting that if $\mv{w}_{k,n}\neq \mv{0}$, then user $k$ is served by RRH $n$; otherwise, user $k$ is not served by RRH $n$. As a result, we can define the user association indicator function $\alpha_{k,n}(\mv{w}_{k,n})$ as follows:
\begin{align}\label{eqn:user association}
\alpha_{k,n}(\mv{w}_{k,n})=\left\{\begin{array}{ll}1, & {\rm if} ~ \mv{w}_{k,n}\neq \mv{0}, \\ 0, & {\rm otherwise}, \end{array} \right. ~~~ k=1,\cdots,K, ~ n=1,\cdots,N.
\end{align}If user $k$ is served by RRH $n$, i.e., $\alpha_{k,n}(\mv{w}_{k,n})=1$, the CU needs to send the digital messages $s_k$ to RRH $n$ over its fronthaul link at a rate of $R_k$, i.e., $B\log_2(1+\gamma_k)$ bps; otherwise, the CU does not need to send $s_k$ to RRH $n$. As a result, the fronthaul capacity constraint of RRH $n$ can be expressed as\footnote{Although the beamforming vectors in general change with wireless channels, since the channel coherence time is assumed to be much longer than the data symbol duration, the number of bits used to send the quantized beamforming vectors is much smaller than that to send the user messages over fronthaul links. As a result, in this paper we assume that each RRH's fronthaul capacity is mainly consumed by the user messages and the quantization for the beamforming vectors sent to each RRH is perfect.}
\begin{align}
\sum\limits_{k=1}^KB\alpha_{k,n}(\mv{w}_{k,n})\log_2(1+\gamma_k)\leq \bar{T}_n, ~~~ n=1,\cdots,N.
\end{align}

\section{Problem Formulation}\label{sec:Problem Formulation}
In this paper, we aim to design the beamforming vectors at all RRHs, i.e., $\mv{w}_{k,n}$'s, which also indicate the user association solution, i.e., $\alpha_{k,n}(\mv{w}_{k,n})$'s according to (\ref{eqn:user association}), to maximize the minimum SINR of all the users subject to each RRH's transmit power constraint as well as its fronthaul capacity constraint. Specifically, we consider the following SINR balancing problem.
\begin{align}\mathrm{(P1)}:~\mathop{\mathtt{Max}}_{\{\mv{w}_{k,n}\},\gamma} & ~~~ \gamma \nonumber  \\
\mathtt {S.t.} ~~ & ~~~ \gamma_k\geq \gamma, ~~~ \forall k, \label{eqn:constraint 1} \\ & ~~~ \sum_{k=1}^K\|\mv{w}_{k,n}\|^2\leq \bar{P}_n, ~~~ \forall n,  \label{eqn:constraint 2} \\ & ~~~  \sum\limits_{k=1}^KB\alpha_{k,n}(\mv{w}_{k,n})\log_2(1+\gamma_k)\leq \bar{T}_n, ~~~ \forall n.  \label{eqn:constraint 3}
\end{align}It is worth noting that without the fronthaul capacity constraints given in (\ref{eqn:constraint 3}), each user should be served by all the RRHs, i.e., $\alpha_{k,n}(\mv{w}_{k,n})=1$, $\forall k,n$, and the optimal beamforming solution can be efficiently obtained \cite{Eldar06}, \cite{Rui10}. Specifically, given any SINR target $\bar{\gamma}$, we need to check whether there exists at least one beamforming solution to satisfy constraints (\ref{eqn:constraint 1}) and (\ref{eqn:constraint 2}). This feasibility problem can be transformed into a second-order cone programming (SOCP) that can be efficiently solved \cite{Boyd04}. Then based on whether such a feasible beamforming solution can be found or not with each given $\bar{\gamma}$, the bisection method \cite{Boyd04} can be applied to let $\bar{\gamma}$ iteratively converge to the optimal SINR solution $\gamma$ to problem (P1).

However, with the newly introduced fronthaul capacity constraints given in (\ref{eqn:constraint 3}), in general each RRH cannot support all the users in the downlink transmission, and as a result, from (\ref{eqn:user association}), for each RRH $n$, only a subset of users are associated with it, for which the corresponding user association function $\alpha_{k,n}(\mv{w}_{k,n})$ and beamforming vector $\mv{w}_{k,n}$ are non-zero. Therefore, the RRH beamforming and user-RRH association need to be jointly optimized in (P1). However, due to the coupling between beamforming and user association given in (\ref{eqn:user association}), problem (P1) is NP-hard and thus difficult to be optimally solved by exhaustive search of all possible user-RRH associations (which has the complexity of $O(2^{NK})$), especially in a network with large number of users and/or RRHs. To the best knowledge of the authors, problem (P1) has not been efficiently solved in the literature.
\section{Proposed Solution}\label{sec:Proposed Solution}
In this section, we propose a new algorithm to efficiently solve problem (P1), which obtains a suboptimal solution in general. First, we have the following proposition.
\begin{proposition}\label{proposition1}
There always exists an optimal solution to problem (P1) such that all the users achieve the same SINR. As a result, problem (P1) is equivalent to the following problem.
\begin{align}\mathrm{(P2)}:~\mathop{\mathtt{Max}}_{\{\mv{w}_{k,n}\},\gamma} & ~~~ \gamma \nonumber  \\
\mathtt {S.t.} ~~ & ~~~ (\ref{eqn:constraint 1}), ~ (\ref{eqn:constraint 2}), \nonumber \\ & ~~~  \sum\limits_{k=1}^KB\alpha_{k,n}(\mv{w}_{k,n})\log_2(1+\gamma)\leq \bar{T}_n, ~~~ \forall n.  \label{eqn:constraint 5}
\end{align}
\end{proposition}

\begin{proof}
Please refer to Appendix \ref{appendix1}.
\end{proof}

Notice that in (P2), comparing the constraint in (\ref{eqn:constraint 5}) with that given by (\ref{eqn:constraint 3}) in (P1), individual user SINR $\gamma_k$ is replaced by the user common SINR $\gamma$ without loss of optimality. As explained in Section \ref{sec:Problem Formulation}, the main difficulty to solve (P2) lies in the coupled RRH beamforming design and user-RRH association through (\ref{eqn:user association}). Hence, in the following, we propose a new method to solve problem (P2) by decoupling the designs of user association and beamforming. In brief, we first allow each user to be served by all the RRHs, i.e., $\alpha_{k,n}(\mv{w}_{k,n})=1$, $\forall k,n$, and then iteratively remove an active user-RRH association with a selected user and RRH pair until the fronthaul capacity constraints of all RRHs are satisfied in (\ref{eqn:constraint 5}).

Specifically, the proposed algorithm iterates as follows until convergence. In the $t$th iteration, with $t\geq 1$, define $\Omega_n^{(t)}$ as the set of users served by RRH $n$. At start, we have $\Omega_n^{(1)}=\mathcal{K}$, $\forall n$. In other words, each user is served by all the RRHs initially. Next, we solve problem (P2) at the $t$th iteration with a given user association solution $\alpha_{k,n}(\mv{w}_{k,n})=0$ if $k\notin \Omega_n^{(t)}$, and $\alpha_{k,n}(\mv{w}_{k,n})=1$ if $k\in \Omega_n^{(t)}$, by solving the following problem.
\begin{align}\mathrm{(P2,t)}:~\mathop{\mathtt{Max}}_{\{\mv{w}_{k,n}\},\gamma} & ~~~ \gamma \nonumber  \\
\mathtt {S.t.} ~~ & ~~~ (\ref{eqn:constraint 1}), ~ (\ref{eqn:constraint 2}), \nonumber \\ & ~~~  B|\Omega_n^{(t)}|\log_2(1+\gamma)\leq \bar{T}_n, ~~~ \forall n, \label{eqn:constraint 7} \\ & ~~~ \|\mv{w}_{k,n}\|^2=0, ~~~ \forall k\notin \Omega_n^{(t)},  \label{eqn:constraint 6}
\end{align}where $|A|$ denotes the cardinality of set $A$. For convenience, we denote problem (P2,t) without fronthaul constraints (\ref{eqn:constraint 7}) as problem (P2-1,t). This problem can be efficiently solved similarly to problem (P1) without the fronthaul constraints (\ref{eqn:constraint 3}) as in \cite{Eldar06} and \cite{Rui10}. Let $\mv{w}_{k,n}^{(1,t)}$'s and $\gamma^{(1,t)}$ denote the optimal solution and optimal value of problem (P2-1,t), respectively.

Then, consider problem (P2,t) only with the fronthaul constraints (\ref{eqn:constraint 7}), which is denoted as problem (P2-2,t). The optimal value of this problem can be easily obtained as
\begin{align}\label{eqn:opt 1}
\gamma^{(2,t)}=\min\limits_{1\leq n \leq N} 2^{\frac{\bar{T}_n}{B|\Omega_n^{(t)}|}}-1.
\end{align}

To summarize, $\gamma^{(1,t)}$ and $\gamma^{(2,t)}$ denote the maximum achievable minimum (max-min) user SINR by only considering the wireless and fronthaul links, respectively, with given user association $\Omega_n^{(t)}$'s. Then, the optimal value of problem (P2,t) is obtained in the following proposition.
\begin{proposition}\label{proposition2}
The optimal value of problem (P2,t) is
\begin{align}\label{eqn:opt 2}
\gamma^{(t)}=\min (\gamma^{(1,t)},\gamma^{(2,t)}).
\end{align}
\end{proposition}

\begin{proof}
Please refer to Appendix \ref{appendix2}.
\end{proof}

In the $t$th iteration of the proposed algorithm, after problem (P2,t) is solved based on Proposition \ref{proposition2}, we check whether the optimal max-min SINR of problem (P2-1,t) satisfies the fronthaul constraints (\ref{eqn:constraint 7}). If so, i.e., $\gamma^{(1,t)}\leq \gamma^{(2,t)}$, the algorithm terminates. Otherwise, if $\gamma^{(1,t)}>\gamma^{(2,t)}$, we will shut down one more wireless link and remove its corresponding fronthaul for a given user and RRH pair, denoted by $(k^{(t)},n^{(t)})$, and then solve problem (P2) with the updated user association solution:
\begin{align}\label{eqn:user association new}
\Omega_n^{(t+1)}=\left\{\begin{array}{ll}\Omega_n^{(t)} \backslash \{k^{(t)}\}, & {\rm if} ~ n= n^{(t)}, \\ \Omega_n^{(t)}, & {\rm otherwise},\end{array}\right.
\end{align}in the $(t+1)$-th iteration, where $A\backslash B$ denotes the set $\{x|x\in A ~ {\rm and} ~ x\notin B\}$. The following proposition ensures the convergence of the above algorithm.

\begin{proposition}\label{proposition3}
The stopping criterion $\gamma^{(1,t)}\leq \gamma^{(2,t)}$ is guaranteed to be satisfied at some iteration, denoted by iteration $t^\ast$. Moreover, before the stopping criterion is satisfied, the obtained max-min SINR of problem (P2,t) is non-decreasing after each iteration, i.e., $\gamma^{(t)}\geq \gamma^{(t-1)}$, $\forall t\leq t^\ast-1$.
\end{proposition}

\begin{proof}
Please refer to Appendix \ref{appendix3}.
\end{proof}

It is worth noting that Proposition \ref{proposition3} only shows that $\gamma^{(t)}\geq \gamma^{(t-1)}$, $\forall t\leq t^\ast-1$. However, in general it is unknown which one between $\gamma^{(t^\ast)}=\gamma^{(1,t^\ast)}$ and $\gamma^{(t^\ast-1)}=\gamma^{(2,t^\ast-1)}$ is greater. As a result, after the algorithm converges, we need to compare the values of $\gamma^{(t^\ast)}$ and $\gamma^{(t^\ast-1)}$ to select the larger one as the converged max-min SINR. Notice that the maximum number of iterations for the above algorithm to converge is $NK$.

The only remaining problem for the proposed algorithm is how to select a wireless link to shut down in each iteration, i.e., how to update $\Omega_n^{(t+1)}$'s from $\Omega_n^{(t)}$'s. According to Proposition \ref{proposition2}, to converge to a max-min SINR as large as possible, it is wise to shut down wireless links in a way such that $\gamma^{(1,t)}$ decreases slowly with $t$, but $\gamma^{(2,t)}$ increases rapidly with $t$. In the following, we propose a practical method to update $\Omega_n^{(t)}$'s based on the above idea.

It can be observed from (\ref{eqn:opt 1}) that to increase the value of $\gamma^{(2,t)}$ in each iteration, one wireless link associated with the RRH with the smallest value of $\bar{T}_n/|\Omega_n^{(t)}|$ should be shut down. Define the set of RRHs with the smallest value of $\bar{T}_n/|\Omega_n^{(t)}|$ as follows:
\begin{align}\label{eqn:rrh 1}
\Psi^{(t)}=\left\{n\left| \frac{\bar{T}_n}{|\Omega_n^{(t)}|}=\min\limits_{1\leq \bar{n} \leq N} \frac{\bar{T}_{\bar{n}}}{|\Omega_{\bar{n}}^{(t)}|}, n=1,\cdots,N.\right.\right\}.
\end{align}Furthermore, define the active wireless links supported by all the RRHs $n\in \Psi^{(t)}$ as follows:
\begin{align}\label{eqn:link 1}
\Phi^{(t)}=\left\{(k,n)\left| n\in \Psi^{(t)}, ~ k\in \Omega_n^{(t)} \right.\right\}.
\end{align}To make $\gamma^{(1,t)}$ decrease slowly in each iteration, we need to solve problem (P2-1,t) $|\Phi^{(t)}|$ times, each time with the additional constraint that one wireless link in $\Phi^{(t)}$ is shut down, and then remove the one that causes the minimum max-min SINR reduction, which is however of high complexity for implementation. In this paper, we propose a low-complexity scheme that shuts down one wireless link based on the existing beamforming solution to problem (P2-1,t), i.e., $\mv{w}_{k,n}^{(1,t)}$'s. Note that with $\mv{w}_{k,n}^{(1,t)}$'s, all users achieve the same SINR. As a result, if one wireless link $(\bar{k},\bar{n})\in \Phi^{(t)}$ is removed, i.e., $\mv{w}_{k,n}=\mv{0}$ if $(k,n)=(\bar{k},\bar{n})$, and $\mv{w}_{k,n}=\mv{w}_{k,n}^{(1,t)}$ otherwise, user $\bar{k}$'s SINR becomes the minimum among the SINRs of all users, which should be maximized. Therefore, we select one wireless link in the set $\Phi^{(t)}$ to shut down based on the following criterion:
\begin{align}\label{eqn:SNR}
(k^{(t)},n^{(t)})=\arg \max \limits_{(\bar{k},\bar{n})\in \Phi^{(t)}} \frac{\sum\limits_{n\neq \bar{n}}\left|\mv{h}_{\bar{k},n}^H \mv{w}_{\bar{k},n}^{(1,t)}\right|^2}{\sum\limits_{j\neq \bar{k}}\left|\sum\limits_{n=1}^N\mv{h}_{\bar{k},n}^H\mv{w}_{j,n}^{(1,t)}\right|^2+\sigma^2}.
\end{align}

\begin{table}[ht]
\begin{center}
\caption{\textbf{Algorithm \ref{table1}}: Proposed Solution for Problem (P2)} \vspace{0.2cm}
 \hrule
\vspace{0.3cm}
\begin{itemize}
\item[a)] Initialize: Set $\Omega_n^{(1)}=\mathcal{K}$, $\forall n$, and $t=1$;
\item[b)] {\bf Repeat}
\begin{itemize}
\item[1)] Solve problem (P2-1,t), obtain the optimal value $\gamma^{(1,t)}$ and optimal beamforming solution $\mv{w}_{k,n}^{(1,t)}$'s;
\item[2)] Solve problem (P2-2,t), obtain the optimal value $\gamma^{(2,t)}$ by (\ref{eqn:opt 1}). If $\gamma^{(1,t)}\geq \gamma^{(2,t)}$, obtain one beamforming solution $\mv{w}_{k,n}^{(2,t)}$'s to achieve $\gamma^{(2,t)}$;
\item[3)] Set the optimal value of problem (P2,t), i.e., $\gamma^{(t)}$, according to (\ref{eqn:opt 2});
\item[4)] If $\gamma^{(1,t)}\leq\gamma^{(2,t)}$, terminate the algorithm; otherwise, update $\Omega_n^{(t+1)}$'s according to (\ref{eqn:user association new}), (\ref{eqn:SNR}) and $t=t+1$;
\end{itemize}
\item[c)] Set the optimal value of problem (P2) as $\gamma^\ast=\max(\gamma^{(t)},\gamma^{(t-1)})$. If $\gamma^\ast=\gamma^{(t)}$, set the optimal beamforming solution as $\mv{w}_{k,n}^\ast=\mv{w}_{k,n}^{(1,t)}$'s; otherwise, set $\mv{w}_{k,n}^\ast=\mv{w}_{k,n}^{(2,t-1)}$'s. Set the optimal user association according to (\ref{eqn:user association}) with given $\mv{w}_{k,n}^\ast$'s.
\end{itemize}
\vspace{0.2cm} \hrule \label{table1} \end{center}
\end{table}

The overall algorithm proposed for problem (P2), denoted by Algorithm \ref{table1}, is summarized in Table \ref{table1}.


\section{Numerical Results}\label{sec:Numerical Results}
In this section, we provide one numerical example to verify our results. In this example, there are $N=5$ RRHs, each equipped with $M=5$ antennas, and $K=15$ users randomly distributed in a circle area of radius $500$m. The bandwidth of the wireless link is $B=10$MHz. The channel vectors are generated from independent Rayleigh fading, while the path loss model of the wireless channel is given as $30.6+36.7\log_{10}(d)$ in dB, where $d$ (in meter) denotes the distance between the user and the RRH. The transmit power constraint for each RRH is $\bar{P}_n=30$dBm, $\forall n$. The power spectral density of the AWGN at each user receiver is assumed to be $-169$dBm/Hz, and the noise figure due to the receiver processing is $7$dB. Moreover, we assume that all the RRHs possess the same fronthaul capacity, i.e., $\bar{T}_n=T$, $\forall n$.

Besides Algorithm \ref{table1}, we also consider the following benchmark schemes for performance comparison.
\begin{itemize}
\item{\bf Benchmark Scheme 1: Minimum Interference Leakage based User Association.} In this scheme, we still perform Algorithm \ref{table1} to solve problem (P2). However, the user selection criterion given in (\ref{eqn:SNR}) is changed to the following:
\begin{align}\label{eqn:SNR1}
(k^{(t)},n^{(t)})=\arg \max \limits_{(k,n)\in \Phi^{(t)}} \sum\limits_{j\neq k} |\mv{h}_{j,n}^H\mv{w}_{k,n}^{(1,t)}|^2.
\end{align}In other words, we shut down the wireless link in $\Phi^{(t)}$ which generates the most significant interference to other users at the $t$th iteration.

\item{\bf Benchmark Scheme 2: Channel-based Greedy User Association.} In this scheme, we first associate each user to its nearest RRH. Then, we iteratively activate one wireless link with the strongest channel power among all inactive wireless links and obtain the optimal beamforming solution to problem (P2) with the updated user association based on Proposition \ref{proposition2}. This procedure is iterated until the obtained max-min SINR is reduced compared to previous iteration.

\item{\bf Benchmark Scheme 3: Conventional Cellular Network.} In this scheme, we consider the conventional cellular network where each user is served by its nearest RRH. With this user association solution, the optimal beamforming solution can be obtained by solving problem (P2-1,t) as in \cite{Eldar06} and \cite{Rui10}.

\end{itemize}


\begin{figure}
\begin{center}
\scalebox{0.5}{\includegraphics*{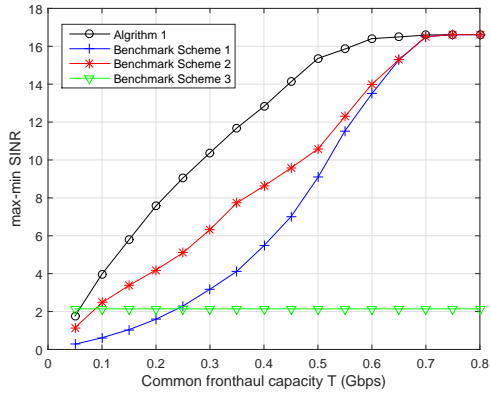}}
\end{center}
\caption{Performance comparison for different schemes versus per-RRH fronthaul capacity.}\label{fig3}
\end{figure}

Fig. \ref{fig3} shows the max-min SINR of all the users achieved by different schemes versus the common fronthaul link capacity $T$, where the performance is averaged over $100$ channel realizations. Note that in Fig. \ref{fig3}, small and large values of $T$ may correspond to the cases of wireless and wired fronthaul links, respectively. It is observed that when $T\leq 0.7$Gbps such that the fronthaul capacity constraints are active, Algorithm \ref{table1} always achieves the best max-min SINR compared to Benchmark Schemes 1-3. It is also observed that with the new selection criterion (\ref{eqn:SNR1}), the max-min SINR achieved by Benchmark Scheme 1 is much lower than that achieved by Algorithm \ref{table1}. This is because the criterion (\ref{eqn:SNR1}) does not consider the influence of shutting down the wireless link $(k^{(t)},n^{(t)})$ on SINR reduction of user $k^{(t)}$, whose SINR is the bottleneck for problem (P2-1,t+1) in the $(t+1)$-th iteration of Algorithm \ref{table1}. As a result, with the user selection given in (\ref{eqn:SNR1}), $\gamma^{(1,t)}$ reduces rapidly as $t$ increases, which is undesired. Moreover, it is observed that there is a performance gap between the max-min SINRs achieved by Benchmark Scheme 2 and Algorithm \ref{table1} since the former does not consider the fronthaul constraint when it selects each new user association. Last, it is observed that Algorithm \ref{table1} and Benchmark Schemes 1-2 all outperform Benchmark Scheme 3 significantly with sufficiently large fronthaul capacity thanks to the joint signal processing gain of C-RAN.


\section{Conclusion}\label{sec:Conclusion}
In this paper, we studied the downlink communication in C-RAN with multi-antenna RRHs and single-antenna users. By jointly optimizing RRH beamforming and user-RRH association, the minimum SINR of all users in C-RAN is maximized subject to each RRH's individual fronthaul capacity constraint. To reduce the complexity of joint optimization, we proposed a new method that decouples the optimization of beamforming and user association, where the number of users served by each RRH is iteratively reduced until the obtained beamforming and user association solution satisfies all the fronthaul capacity constraints. The convergence of the proposed algorithm was proved. We also showed by simulation that the proposed algorithm outperforms other heuristic solutions. It is worth noting that our proposed new method can be generally applied to other downlink transmission optimization problems in C-RAN or other wireless networks subject to individual fronthaul/backhaul constraints at the RRHs/BSs.

\begin{appendix}
\subsection{Proof of Proposition \ref{proposition1}}\label{appendix1}
Suppose that with the optimal beamforming solution of (P1), the SINRs of the users are denoted by $\mv{\gamma}'=[\gamma_1',\cdots,\gamma_K']^T$. Consider another SINR vector for the users denoted by $\mv{\gamma}''=[\gamma'',\cdots,\gamma'']^T$, where $\gamma''=\min_{1\leq k \leq K} \gamma_k'$. First, the objective value of problem (P1) is not changed with the new users' SINR vector $\mv{\gamma}''$. Next, without the fronthaul constraints given in (\ref{eqn:constraint 3}), there must exist a beamforming solution, denoted by $\{\mv{w}_{k,n}''\}$, such that $\gamma''$ is achievable by all users since $\gamma''\leq \gamma_k'$, $\forall k$. Similarly, $\{\mv{w}_{k,n}''\}$ also satisfies the fronthaul constraints (\ref{eqn:constraint 3}). As a result, given any beamforming solution to problem (P1), we can always find another solution such that all the users achieve the same SINR. Proposition \ref{proposition1} is thus proved.

\subsection{Proof of Proposition \ref{proposition2}}\label{appendix2}
The optimal value of problem (P2,t) must be upper-bounded by the optimal values of its sub-problems (P2-1,t) and (P2-2,t), i.e., $\gamma^{(t)}\leq \min(\gamma^{(1,t)},\gamma^{(2,t)})$. In the following, we show that the equality is always achievable. First, consider the case when $\gamma^{(1,t)}\geq \gamma^{(2,t)}$. Since $\gamma^{(2,t)}$ is no larger than the optimal value of problem (P2-1,t), there must exist one beamforming solution, denoted by $\mv{w}_{k,n}^{(2,t)}$'s, such that the SINR target $\gamma^{(2,t)}$ is simultaneously achieved by all the users over the wireless links. As a result, in this case the optimal value of problem (P2,t) is $\gamma^{(t)}=\gamma^{(2,t)}=\min(\gamma^{(1,t)},\gamma^{(2,t)})$, which is achieved by the beamforming solution $\mv{w}_{k,n}^{(2,t)}$'s. Next, consider the case when $\gamma^{(1,t)}<\gamma^{(2,t)}$. In this case, the fronthaul capacity constraints (\ref{eqn:constraint 7}) are satisfied even if all the users' SINRs are equal to $\gamma^{(1,t)}$. As a result, the optimal value of problem (P2,t) is $\gamma^{(t)}=\gamma^{(1,t)}=\min(\gamma^{(1,t)},\gamma^{(2,t)})$, which is achieved by the beamforming solution $\mv{w}_{k,n}^{(1,t)}$'s. Proposition \ref{proposition2} is thus proved.

\subsection{Proof of Proposition \ref{proposition3}}\label{appendix3}
First, it can be shown that if we shut down one wireless link, the max-min SINR over the wireless links is non-increasing. As a result, it follows that the optimal value of problem (P2-1,t) is non-increasing with $t$. Second, since the number of users served by each RRH is non-increasing with $t$, according to (\ref{eqn:opt 1}) the optimal value of problem (P2-2,t) is non-decreasing with $t$. As a result, the gap between $\gamma^{(1,t)}$ and $\gamma^{(2,t)}$ will be non-increasing as $t$ increases. Moreover, if all the wireless links are shut down at some iteration $\hat{t}$, then it follows that $\gamma^{(1,\hat{t})}<\gamma^{(2,\hat{t})}$. As a result, there must exist one $t^\ast<\hat{t}$ such that $\gamma^{(1,t)}\geq \gamma^{(2,t)}$ when $t<t^\ast$, but $\gamma^{(1,t)}<\gamma^{(2,t)}$ when $t\geq t^\ast$. The first part of Proposition \ref{proposition3} is thus proved.

Next, according to Proposition \ref{proposition1}, before the stopping criterion is satisfied, we have $\gamma^{(t)}=\gamma^{(2,t)}$, which is non-decreasing with $t$. As a result, the second part of Proposition \ref{proposition3} is proved.
\end{appendix}

\end{document}